\documentclass[letterpaper, 10 pt, conference]{ieeeconf}

\IEEEoverridecommandlockouts

\usepackage[utf8]{inputenc}
\usepackage[english]{babel}

\usepackage{graphicx}
\graphicspath{{Figures/}}

\usepackage{mathtools, amssymb, amsthm, bm}

\newtheorem{lemma}{Lemma}
\newtheorem{proposition}{Proposition}
\theoremstyle{definition}
\newtheorem{definition}{Definition}
\theoremstyle{remark}
\newtheorem{remark}{Remark}

\let\labelindent\relax
\usepackage{enumerate}
\usepackage{enumitem}
\usepackage{array}
\usepackage{cite}
\usepackage{balance}
\usepackage{multirow}
\usepackage{booktabs}

\usepackage{tikz}
\usetikzlibrary{shapes, arrows, positioning, calc, intersections}
\tikzstyle{block} = [draw, thick, rectangle, fill=white!20, minimum height=3em, minimum width=6em]
\tikzstyle{square} = [draw, thick, rectangle, fill=white!20, minimum height = 3em, minimum width = 3em]
\tikzstyle{dotbox} = [draw, dashed, thick, rectangle, fill=white!20, minimum height = 3em, minimum width = 3em]
\tikzstyle{coord} = [coordinate]

\newcommand{\sign}{\mathrm{sign}}

\title{\LARGE \bf How Much Reserve Fuel: Quantifying the Maximal Energy Cost\\of System Disturbances}

\author{Ram Padmanabhan$^{*,1,4}$, Craig Bakker$^2$, Siddharth Abhijit Dinkar$^1$, and Melkior Ornik$^{3,4}$
\thanks{This work was supported by the Resilience through Data-driven Intelligently-Designed Control (RD2C) Initiative under the Laboratory Directed Research and Development (LDRD) Program at Pacific Northwest National Laboratory (PNNL), Air Force Office of Scientific Research grant FA9550-23-1-0131 and NASA University Leadership Initiative grant 80NSSC22M0070.}
\thanks{$^*$Corresponding Author. Email: {ramp3@illinois.edu}}
\thanks{$^1$Department of Electrical and Computer Engineering, University of Illinois Urbana-Champaign, Urbana, IL, USA.}%
\thanks{$^2$National Security Directorate, Pacific Northwest National Laboratory, Richland, WA, USA.}%
\thanks{$^3$Department of Aerospace Engineering, University of Illinois Urbana-Champaign, Urbana, IL, USA.}%
\thanks{$^4$Coordinated Science Laboratory, University of Illinois Urbana-Champaign, Urbana, IL, USA.}
}

\begin{document}

\maketitle

\begin{abstract}
Motivated by the design question of additional fuel needed to complete a task in an uncertain environment, this paper introduces metrics to quantify the maximal additional energy used by a control system in the presence of bounded disturbances when compared to a nominal, disturbance-free system. In particular, we consider the task of finite-time stabilization for a linear time-invariant system. We first derive the nominal energy required to achieve this task in a disturbance-free system, and then the worst-case energy over all feasible disturbances. The latter leads to an optimal control problem with a least-squares solution, and then an infinite-dimensional optimization problem where we derive an upper bound on the solution. The comparison of these energies is accomplished using additive and multiplicative metrics, and we derive analytical bounds on these metrics. Simulation examples on an ADMIRE fighter jet model demonstrate the practicability of these metrics, and their variation with the task hardness, a combination of the distance of the initial condition from the origin and the task completion time.
\end{abstract}

\section{Introduction} \label{sec:Introduction}
Real-world control systems are almost always subjected to external disturbances and noise in uncertain environments. When operating in safety- or mission-critical scenarios, these perturbations can prevent a system from achieving a specified performance objective, such as reachability or safety \cite{Yan23}. Furthermore, attempting to directly compensate for these perturbations while still satisfying a performance objective requires additional fuel in practical systems \cite{Na20}, which is closely related to additional energy in the control signal.

The setting of robust control theory \cite{RC1, RC2} focuses on designing controllers for systems that operate in environments with external perturbations, regardless of the characteristics of these perturbations, guaranteeing the notion of \emph{strong reachability}. The design of such control laws has been widely studied \cite{BR71, RKML06}, and is motivated by the fact that these external disturbances are not known \emph{a priori}. A variant of the robust control setting considers the \emph{full information} problem, where the controller has complete knowledge of plant states and disturbances \cite{FI1, FI2}. The design of optimal controllers in such uncertain environments typically focuses on minimizing robustness metrics such as the $\mathcal{H}_2$ or $\mathcal{H}_{\infty}$ norms of the closed-loop system \cite{RC1}. These metrics measure the variance or energy amplification from the disturbance to the state or output. In contrast, we consider the control energy expended to reach a target set. Compensating for system disturbances generally requires additional control energy compared to the ideal, disturbance-free setting. This detail is particularly relevant in vehicular systems \cite{Na20}, where control energy is closely related to the amount of expended fuel. The design of controllers for fuel economy is a well-researched topic, with studies on marine vehicles \cite{Ship12}, aircraft \cite{Elias22}, autonomous vehicles \cite{MS16} and hybrid vehicles \cite{HV1, HV2}. When subjected to external perturbations, such systems must consider the additional fuel expended to compensate for these perturbations while achieving a target.

Based on these considerations, in this paper we introduce metrics to quantify the maximal additional control energy used by a control system to achieve a target, in the presence of bounded disturbances. The metrics we introduce are closely related to the notion of \emph{quantitative resilience} introduced by Bouvier \emph{et al.}\cite{BXO21, BO23}. Quantitative resilience refers to the maximal ratio of minimal reach times required to reach a target for a nominal system and a system which partially loses control authority over actuators. Our metrics consider minimal control energy instead of minimal reach time and disturbances instead of malfunctioning actuators. Considering the task of finite-time stabilization, we compare the control energy required to achieve this task in an nominal, disturbance-free linear time-invariant (LTI) system to the worst-case energy required in a system subjected to bounded disturbances. We note that finding this worst-case energy only requires knowledge of the bounds on the disturbance, and this constitutes our main results. Our first result derives an upper bound on this worst-case energy over all feasible disturbances, based on an optimal control problem with a least-squares solution and then an infinite-dimensional optimization problem. Next, we define two metrics that encapsulate the \emph{cost of disturbance}, in terms of the additional control energy required to achieve the task in a disturbed system compared to the ideal system. We derive bounds on these metrics and illustrate their applicability on an ADMIRE fighter jet model \cite{SDRA}, demonstrating their variation with the initial condition distance from the origin as well as a metric quantifying the difficulty of completing the task. 

The remainder of this paper is organized as follows. Section \ref{sec:Preliminaries} provides fundamental notation and useful mathematical facts, and introduces our problem formulation, defining the control energies we are interested in. In Section \ref{sec:Energies}, we derive the optimal control signals and control energies based on these definitions, and obtain an upper bound on the worst-case control energy over all disturbances. We introduce the cost of disturbance metrics in Section \ref{sec:Metrics}, and derive bounds on these quantities. Finally, in Section \ref{sec:Example} we illustrate the performance of the control laws and the practicality of the metrics we define. In particular, we show that for tasks with a large distance of initial condition from the origin or a small final time, our metrics are useful in determining the additional energy required to achieve the given task.

\section{Preliminaries} \label{sec:Preliminaries}
\subsection{Notation and Facts}
The $\sign(\cdot)$ function is defined on $z \in \mathbb{R}$ as $\sign(z) = z/|z| \in \{-1, +1\}$ if $z \neq 0$, with $\sign(0) = 0$. If $z \in \mathbb{R}^n$, the $\sign(\cdot)$ function operates elementwise on $z$. The $p$-norm of a vector $x \in \mathbb{R}^n$ is defined as $\|x\|_p \coloneqq \left(\sum_{i=1}^{n}|x_i|^p\right)^{1/p}$, with $\|x\|_{\infty} \coloneqq \max_i |x_i|$. For a matrix $L \in \mathbb{R}^{p\times q}$ with entries indexed $l_{ij}$, define the induced matrix norms $\|L\|_1 \coloneqq \max_j \sum_{i=1}^{p} |l_{ij}|$ and $\|L\|_{\infty} \coloneqq \max_i \sum_{j=1}^{q} |l_{ij}|$. The sub-multiplicative property of matrix norms states that for any two matrices $M$ and $N$ such that their product $MN$ can be defined, $\|MN\| \leq \|M\|\|N\|$, where $\|\cdot\|$ is any induced norm. The minimum and maximum eigenvalues of a symmetric, positive-definite matrix $P \in \mathbb{R}^{n\times n}$ are denoted $\lambda_{\min}(P)$ and $\lambda_{\max}(P)$. For any vector $x \in \mathbb{R}^n$, the inequality $\lambda_{\min}(P)\|x\|_{2}^{2} \leq x^TPx \leq \lambda_{\max}(P)\|x\|_{2}^{2}$ holds. 

For a continuous function $u:[0, t_f]\to \mathbb{R}^p$, the $\mathcal{L}_2$ norm is defined as $\|u\|_{\mathcal{L}_2} \coloneqq \sqrt{\int_{t=0}^{t_f}\|u(t)\|_{2}^{2}~\mathrm{d}t}$. For a convex function $\varphi:\mathbb{R}^p\to\mathbb{R}$, Jensen's inequality \cite[Chapter 1]{NPBook} states $\varphi\left(\frac{1}{t_f}\int_{0}^{t_f}u(t)\mathrm{d}t\right) \leq \frac{1}{t_f}\int_{0}^{t_f}\varphi(u(t))\mathrm{d}t$. If $\varphi(\cdot) = \|\cdot\|$ is any norm, Jensen's inequality reduces to $\left\|\int_{0}^{t_f}u(t)\mathrm{d}t\right\| \leq \int_{0}^{t_f}\|u(t)\|\mathrm{d}t$, since all norms are convex and absolutely homogeneous.

\subsection{Problem Formulation}
Consider a linear time-invariant system affected by a bounded disturbance:
\begin{equation} \label{eq:Disturbed}
	\dot{x}(t) = Ax(t) + Bu(t) + w(t), ~~x(0) = x_0 \neq 0,
\end{equation}
where $x(t) \in \mathbb{R}^n$, $u(t) \in \mathbb{R}^p$, $w(t) \in \mathbb{R}^n$ and the matrices $A$ and $B$ have appropriate dimensions. The pair $(A, B)$ is controllable, and the disturbance $w(t)$ is pointwise bounded, such that $\|w(t)\|_{\infty} \leq \overline{w}$ for all $t$ with $\overline{w} > 0$. Let $\mathcal{W}$ denote the space of all such pointwise bounded functions. Our objective is the task of finite-time stabilization, achieving $x(t_f) = 0$ from an initial condition $x(0) = x_0 \neq 0$, for a given final time $t_f$, based on design specifications. The space of inputs $\mathcal{U}$ consists of all square-integrable functions over the interval $[0, t_f]$ such that the solution to \eqref{eq:Disturbed} exists \cite{Khalil}.

We note that exactly achieving this finite-time specification requires \emph{a priori} knowledge of $w(t)$ in order to compensate for its effect using the control $u(t)$. This is further discussed in Lemma~\ref{lem:uD2}. Our focus is not on designing this control, but on quantifying the worst-case additional control energy required to achieve our specification, given that $\overline{w}$ is known. In this setting, the nominal, disturbance-free model
\begin{equation} \label{eq:Nominal}
	\dot{x}(t) = Ax(t) + Bu(t), ~~x(0) = x_0 \neq 0,
\end{equation}
is considered as the baseline. In both the actual \eqref{eq:Disturbed} and nominal \eqref{eq:Nominal} dynamics, we are interested in the minimal energy required to achieve finite-time stabilization. To this end, we make the following definitions:

\begin{definition}[Nominal Energy] \label{def:Nominal}
The \emph{nominal energy} is the minimum energy in the control signal required to achieve $x(t_f) = 0$ from $x(0) = x_0$ following the nominal, disturbance-free dynamics \eqref{eq:Nominal}:
\begin{equation} \label{eq:EN_def}
E_{N}^{*}(x_0, t_f) \coloneqq \inf_u \left\{\|u\|_{\mathcal{L}_{2}}^{2}~\mathrm{ s.t. }~x(t_f) = 0~\mathrm{ using ~\eqref{eq:Nominal} }\right\}.
\end{equation}
\end{definition}

\begin{definition}[Disturbed Energy] \label{def:Disturbed}
The \emph{disturbed energy} is the minimal energy in the control signal required to achieve $x(t_f) = 0$ from $x(0) = x_0$ following the actual dynamics \eqref{eq:Disturbed}, when $w(t)$ is chosen to maximize this quantity:
\begin{equation} \label{eq:ED_def}
E_{D}^{*}(x_0, t_f) \coloneqq \sup_{w \in \mathcal{W}} \left\{\inf_u \left\{\|u\|_{\mathcal{L}_{2}}^{2}~\mathrm{ s.t. }~x(t_f) = 0~\mathrm{ using ~\eqref{eq:Disturbed} }\right\} \right\}.
\end{equation}
In other words, the disturbed energy refers to the \emph{worst-case minimum energy}.
\end{definition}

We note that Bouvier and Ornik \cite{BO20, BO22b} discuss finite-time reachability under energy constraints on the input. However, this is in the context of losing control authority over actuators, whereas we consider external disturbances. Further, they consider energy bounds on the malfunctioning actuator inputs, whereas we consider amplitude bounds on the disturbance. We also introduce metrics to quantify the increased energy required under external disturbances in Section \ref{sec:Metrics}, while this is not addressed in the above papers.

\section{Nominal and Disturbed Energies} \label{sec:Energies}
In this section, we discuss the minimum-energy control signals to achieve finite-time stabilization in both the actual \eqref{eq:Disturbed} and nominal \eqref{eq:Nominal} systems. Using these signals, we obtain the nominal and disturbed energies defined in \eqref{eq:EN_def} and \eqref{eq:ED_def}. 

\begin{lemma}[Nominal Energy]{\cite[Chapter 3]{LVS}} \label{lem:EN}
The minimum-energy control signal that achieves finite-time stabilization in the nominal system \eqref{eq:Nominal} is
\begin{equation} \label{eq:Nominal_u}
u_N(t) = -B^Te^{A^T(t_f - t)}W_{B}^{-1}e^{At_f}x_0
\end{equation}
for $t \in [0, t_f]$, where $W_B$ is the finite-horizon controllability Gramian:
\begin{equation} \label{eq:WB}
W_B \coloneqq \int_{0}^{t_f} e^{At}BB^Te^{A^Tt}~\mathrm{d}t \succ 0.
\end{equation}
The corresponding minimum energy is given by
\begin{equation} \label{eq:EN}
	E_{N}^{*}(x_0, t_f) = \|u_N\|_{\mathcal{L}_2}^{2} = x_{0}^{T}e^{A^Tt_f}W_{B}^{-1}e^{At_f}x_0.
\end{equation}
\end{lemma}

The well-known result in Lemma \ref{lem:EN} follows from a least-squares argument in function spaces. We can similarly obtain an expression for the inner infimum in \eqref{eq:ED_def}.

\begin{lemma} \label{lem:uD2}
The minimum-energy control signal that achieves finite-time stabilization in the actual system \eqref{eq:Disturbed} is
\begin{equation} \label{eq:Disturbed_u}
u_D(t) = -B^Te^{A^T(t_f - t)}W_{B}^{-1}\Bigg[e^{At_f}x_0 + {\mathcal{R}(w, t_f)}\Bigg]
\end{equation}
for $t \in [0, t_f]$, where $\mathcal{R}(w, t_f) = {\int_{0}^{t_f} e^{A(t_f-\tau)}w(\tau)~\mathrm{d}\tau}$ and $W_B$ is defined in \eqref{eq:WB}. The energy in this control signal is given by
\begin{align}
	&\|u_D\|_{\mathcal{L}_2}^2 = \int_{0}^{t_f}u_{D}^{T}(t)u_D(t)~\mathrm{d}t \nonumber \\
	&\hspace{0.2em}= \left[e^{At_f}x_0 + \mathcal{R}(w, t_f)\right]^T W_{B}^{-1} \left[e^{At_f}x_0 + \mathcal{R}(w, t_f)\right]. \label{eq:uD2}
\end{align}
\end{lemma}

\begin{proof}
The proof uses a similar least-squares argument to that of Lemma \ref{lem:EN}, and we provide a brief description here. Note that the trajectory of the state in \eqref{eq:Disturbed} is given by:
$$
x(t) = e^{At}x_0 + \int_{0}^{t} e^{A(t - \tau)}Bu(\tau)~\mathrm{d}\tau + \underbrace{\int_{0}^{t} e^{A(t-\tau)}w(\tau)~\mathrm{d}\tau}_{\mathcal{R}(w, t)},
$$
and this follows from similar arguments to the solution of a linear control system \cite[Section 4.4]{Khalil}. Setting $t = t_f$ and $x(t_f) = 0$, the minimum-energy control problem is written as:
\begin{subequations}
\begin{align}
&\inf_u \int_{0}^{t_f} u^T(t) u(t)~\mathrm{d}t \label{eq:P1} \\
&~\text{s.t. } -e^{At_f}x_0 - \mathcal{R}(w, t_f) = \int_{0}^{t_f} e^{A(t_f - \tau)}Bu(\tau)~\mathrm{d}\tau. \label{eq:C1}
\end{align}
\end{subequations}
From \eqref{eq:C1}, it is easy to see that the optimal solution must depend on $\mathcal{R}(w, t_f)$, and thus requires \emph{a priori} knowledge of $w(t)$. The above problem is an infinite-dimensional counterpart of the underdetermined Euclidean least-squares problem:
\begin{align*}
\min_y \|y\|_{2}^{2} ~~\text{ s.t. } z = My,
\end{align*}
whose solution is given by $y^* = M^T(MM^T)^{-1}z$, with $\|y^*\|_{2}^{2} = z^T(MM^T)^{-1}z$. In the infinite-dimensional case, the transpose operator is written as an \emph{adjoint} operator. Comparing the constraints above, we note that the right-hand-side in \eqref{eq:C1} is a linear operator $\mathcal{M}:\mathcal{U}\to\mathbb{R}^n$ on the function $u(t)$, and $z = -e^{At_f}x_0 - \mathcal{R}(w, t_f)$. Let $\mathcal{M}^*:\mathbb{R}^n\to\mathcal{U}$ denote its adjoint operator, which satisfies $\langle \overline{z}, \mathcal{M}(u)\rangle = \langle \mathcal{M}^*(\overline{z}), u\rangle$ for all $\overline{z} \in \mathbb{R}^n$, where the notation $\langle \cdot,\cdot\rangle$ refers to the inner product in the space of the arguments. Then,
\begin{align*}
\langle \overline{z}, \mathcal{M}(u)\rangle &= \left[\int_{0}^{t_f} e^{A(t_f - \tau)}Bu(\tau)~\mathrm{d}\tau\right]^T\overline{z} \\
&= \int_{0}^{t_f} u^T(t)B^Te^{A^T(t_f-\tau)}\overline{z}~\mathrm{d}\tau = \langle \mathcal{M}^*(\overline{z}), u\rangle,
\end{align*}
and thus $\mathcal{M}^*(\overline{z}) = B^Te^{A^T(t_f-t)}\overline{z}$. From the finite-dimensional solution, note that the optimal control can be written as $u_D(t) = \mathcal{M}^*\left(\mathcal{MM}^*\right)^{-1}z$. Thus,
$$
u_D(t) = -B^Te^{A^T(t_f - t)}W_{B}^{-1}\left[e^{At_f}x_0 + \mathcal{R}(w, t_f)\right],
$$
\begin{align*}
	&\text{and }~\|u_D\|_{\mathcal{L}_2}^2 \\
	&\hspace{0.2em}= \left[e^{At_f}x_0 + \mathcal{R}(w, t_f)\right]^T W_{B}^{-1} \left[e^{At_f}x_0 + \mathcal{R}(w, t_f)\right].
\end{align*}
\end{proof}	

Of course, when $w(t) \equiv 0$, we have $\mathcal{R}(w, t_f) = 0$ and $u_D(t) = u_N(t)$. Note that $\|u_D\|_{\mathcal{L}_2}^2$ depends on a specific $w(t)$ and is not the disturbed energy in Definition \ref{def:Disturbed}. The disturbed energy is given by the following result.

\begin{proposition}[Disturbed Energy] \label{prop:ED}
The disturbed energy $E_{D}^{*}(x_0, t_f)$ is bounded above as follows:
\begin{align} \label{eq:ED}
E_{D}^{*}(x_0, t_f) &= \sup_{w \in \mathcal{W}} \|u_D\|_{\mathcal{L}_2}^2 \nonumber \\
&\hspace{-1cm}\leq E_{N}^{*}(x_0, t_f) + 2\overline{q}\|\Lambda U^Te^{At_f}x_0\|_1 + \overline{q}^2\sum_{i = 1}^{n} \lambda_i,
\end{align}
where $W_{B}^{-1} = U\Lambda U^T$ is the spectral decomposition of $W_{B}^{-1}$, $\Lambda = \mathrm{diag}\{\lambda_1, \ldots, \lambda_n\} \succ 0$ and $\overline{q} = \overline{w}~\|U\|_1 \int_{0}^{t_f} \left\|e^{A(t_f - t)}\right\|_{\infty} \mathrm{d}t.$
\end{proposition}

\begin{proof}
We require the solution $E_{D}^{*}(x_0, t_f)$ of the infinite-dimensional problem
\begin{equation} \label{eq:P2}
\sup_{w \in \mathcal{W}} \left[e^{At_f}x_0 + \mathcal{R}(w, t_f)\right]^T W_{B}^{-1} \left[e^{At_f}x_0 + \mathcal{R}(w, t_f)\right].
\end{equation}
Let $v = \mathcal{R}(w, t_f) = \int_{0}^{t_f} e^{A(t_f-t)}w(t)\mathrm{d}t \in \mathbb{R}^n$, and note that
\begin{align}
\|v\|_{\infty} &= \left\|\int_{0}^{t_f} e^{A(t_f-t)}w(t)\mathrm{d}t \right\|_{\infty} \nonumber \\
&\hspace{-1.2cm}\leq \int_{0}^{t_f} \left\|e^{A(t_f-t)}w(t)\right\|_{\infty}\mathrm{d}t \leq \underbrace{\overline{w} \int_{0}^{t_f} \left\|e^{A(t_f-t)}\right\|_{\infty}\mathrm{d}t}_{\overline{v}},
\end{align}
where we use Jensen's inequality and the sub-multiplicative property of norms. Here, $\overline{v}$ is a constant that is independent of the disturbance $w(t)$, and depends only on the bound $\overline{w}$. Then, \eqref{eq:P2} implies
\begin{align}
E_{D}^{*}(x_0, t_f) &\leq E_{N}^{*}(x_0, t_f) + \nonumber \\
&\sup_{\|v\|_{\infty} \leq \overline{v}} \left[2v^T W_{B}^{-1} e^{At_f}x_0 + v^T W_{B}^{-1} v\right], \label{eq:ED_v}
\end{align}
where we now have to solve a finite-dimensional problem. Since $W_B$ is symmetric and positive definite, let $W_{B}^{-1} = U\Lambda U^T$ be the spectral decomposition of $W_{B}^{-1}$, where $\Lambda = \mathrm{diag}\{\lambda_1, \ldots, \lambda_n\} \succ 0$ and $U$ is an orthogonal matrix consisting of the eigenvectors of $W_{B}^{-1}$. Let $q = U^Tv \in \mathbb{R}^n$. Then, $\|q\|_{\infty} \leq \overline{q} \coloneqq \|U\|_1 \overline{v}$, and \eqref{eq:ED_v} implies
\begin{equation} \label{eq:P3}
E_{D}^{*}(x_0, t_f) \leq E_{N}^{*}(x_0, t_f) + \sup_{\|q\|_{\infty} \leq \overline{q}} \left[2q^T p + q^T\Lambda q\right],
\end{equation}
where $p = \Lambda U^Te^{At_f}x_0 \in \mathbb{R}^n$. Note that $\sup_{\|q\|_{\infty} \leq \overline{q}}~ q^Tp = \overline{q}\|p\|_1,$ using the fact that the $1$-norm is dual to the $\infty$-norm \cite[Chapter 5]{CVX}. The optimal $q$ for this problem is $q^* = \overline{q}~\sign(p)$, where the $\sign(\cdot)$ function operates elementwise on the vector $p$. Note that $q^* = \overline{q}~\sign(p)$ also maximizes the term $q^T\Lambda q$ in \eqref{eq:P3} by forcing $q_{i}^{2} = \overline{q}^2$ for each element $q_i$ of $q$, since $\lambda_i > 0$ for $i = 1, \ldots, n$. It then follows that:
\begin{align*}
E_{D}^{*}(x_0, t_f) &\leq E_{N}^{*}(x_0, t_f) + 2\overline{q}\|\Lambda U^Te^{At_f}x_0\|_1 + \overline{q}^2\sum_{i = 1}^{n} \lambda_i.
\end{align*}
\end{proof}

\begin{remark}
We note that the optimal control signal \eqref{eq:Disturbed_u} and its minimum energy \eqref{eq:uD2} depend directly on $w(t)$, whereas the upper bound $\overline{E}_{D}(x_0, t_f)$ in \eqref{eq:ED} depends only on $\overline{w}$. 
\end{remark}

The upper bound in \eqref{eq:ED} may be quite conservative, particularly due to the use of Jensen's inequality and the sub-multiplicative property of norms. Obtaining improved bounds is an important avenue for future work. However, we show in Section \ref{sec:Example} that this upper bound is a more accurate approximation for $\|u_D\|_{\mathcal{L}_2}^{2}$ as $t_f$ becomes smaller. We now introduce two metrics to quantify the larger value of the disturbed energy \eqref{eq:ED} compared to the nominal energy \eqref{eq:EN}.

\section{Cost of Disturbance Metrics} \label{sec:Metrics}
Recall that we aim to quantify the additional control energy required to achieve finite-time stabilization under the actual dynamics \eqref{eq:Disturbed}, compared to the nominal dynamics \eqref{eq:Nominal}. We now use the expressions for the nominal \eqref{eq:EN} and disturbed energies \eqref{eq:ED} obtained in the previous section, to define two metrics. First, define an \emph{additive metric}:
\begin{equation} \label{eq:Additive}
	r_A(t_f) \coloneqq \sup_{\|x_0\|_2 \leq R} E_{D}^{*}(x_0, t_f) - E_{N}^{*}(x_0, t_f),
\end{equation}
which is directly a measure of how much additional energy is required. In particular, the actual system \eqref{eq:Disturbed} uses at most $r_A(t_f)$ more energy than the nominal system \eqref{eq:Nominal} to achieve finite-time stabilization in time $t_f$, assuming optimal control signals are used. The quantity $R > 0$ refers to a limit on the distance of the initial condition from the origin, and the constraint $\|x_0\|_2 \leq R$ is required to ensure that $r_A(t_f)$ is bounded. We then have the following result.

\begin{proposition}[Additive Metric] \label{prop:rA}
The additive metric $r_A(t_f)$ is bounded above as follows:
\begin{equation} \label{eq:rA_bound}
	r_A(t_f) \leq c + \gamma R\sqrt{n},
\end{equation}
where $c = \overline{q}^2\sum_{i = 1}^{n} \lambda_i$, $\gamma = 2\overline{q}\|\Lambda U^Te^{At_f}\|_1$ and $n$ is the dimension of $x_0$.
\end{proposition}

\begin{proof}
The proof follows from \eqref{eq:ED}, the sub-multiplicative property of norms and the inequality $\|x_0\|_1 \leq \sqrt{n}\|x_0\|_2$. 
\end{proof}

Similarly, we define a \emph{multiplicative metric}:
\begin{equation} \label{eq:Multiplicative}
	r_M(t_f) \coloneqq \inf_{\|x_0\|_2 \geq R} ~\frac{E_{N}^{*}(x_0, t_f)}{E_{D}^{*}(x_0, t_f)},
\end{equation}
providing a multiplicative measure of the increased energy required. For instance, if $r_M(t_f) \geq 1/3$, then \emph{at most} 3 times the control energy is required to achieve finite-time stabilization in time $t_f$, assuming optimal control signals are used. The constraint $\|x_0\|_2 \geq R$ ensures that this metric is non-trivial for initial conditions arbitrarily close to the origin. We also note that $r_M(t_f)$ is bounded above by 1. The following result then holds.

\begin{proposition}[Multiplicative Metric] \label{prop:rM}
The multiplicative metric $r_M(t_f)$ is bounded below as follows:
\begin{equation} \label{eq:rM_bound}
r_M(t_f) \geq \frac{lR^2}{lR^2 + \gamma R\sqrt{n} + c},
\end{equation}
where $c = \overline{q}^2\sum_{i = 1}^{n} \lambda_i$, $\gamma = 2\overline{q}\|\Lambda U^Te^{At_f}\|_1$, $l = \lambda_{\min}\left(e^{A^Tt_f}W_{B}^{-1}e^{At_f}\right)$ and $n$ is the dimension of $x_0$.
\end{proposition}

\begin{proof}
Using \eqref{eq:ED}, we have
$$
r_M(t_f) \geq \frac{1}{1 + \sup_{\|x_0\|_2 \geq R} \frac{\gamma\|x_0\|_1 + c}{x_{0}^{T}e^{A^Tt_f}W_{B}^{-1}e^{At_f}x_0}}.
$$
We also have
\begin{align*}
&\sup_{\|x_0\|_2 \geq R} \frac{\gamma\|x_0\|_1 + c}{x_{0}^{T}e^{A^Tt_f}W_{B}^{-1}e^{At_f}x_0} \nonumber \\
\leq &\sup_{\|x_0\|_2 \geq R} \frac{\gamma\sqrt{n}\|x_0\|_2 + c}{l\|x_0\|_{2}^{2}} = \frac{\gamma\sqrt{n}}{lR} + \frac{c}{lR^2},
\end{align*}
where $l = \lambda_{\min}\left(e^{A^Tt_f}W_{B}^{-1}e^{At_f}\right) > 0$, since $W_B$ is positive definite and $e^{At_f}$ is non-singular. We use the inequalities $\|x_0\|_1 \leq \sqrt{n}\|x_0\|_2$ and $x^TPx \geq \lambda_{\min}(P)\|x\|_{2}^{2}$ for a positive-definite matrix $P$, and the result follows.
\end{proof}

From \eqref{eq:rM_bound}, we note that as the distance from the origin $R$ increases, the lower bound on the metric $r_M(t_f)$ also increases and converges to $1$ as $R \to \infty$. This behavior indicates that for initial conditions farther from the origin, the disturbed energy exceeds the nominal energy by a \emph{relatively} smaller amount, than for initial conditions closer to the origin. Intuitively, an initial condition farther from the origin already requires a significant amount of energy to achieve the specification for the nominal system. The additional energy required to compensate the external disturbance $w(t)$ is thus relatively less impactful. 

For large distances of the initial condition from the origin, the nominal and disturbed energies are large, and their difference may be large even if they have the same order of magnitude. Thus, an additive metric does not represent the cost of disturbance accurately, but a multiplicative metric does. Conversely, for small distances of the initial condition from the origin, the nominal and disturbed energies are small, and an additive metric represents the cost of disturbance more accurately. This feature is closely related to the notions of relative and absolute error in numerical analysis \cite{BT}.

\begin{figure}[!t]
	\centering
	\includegraphics[width = 0.4\textwidth]{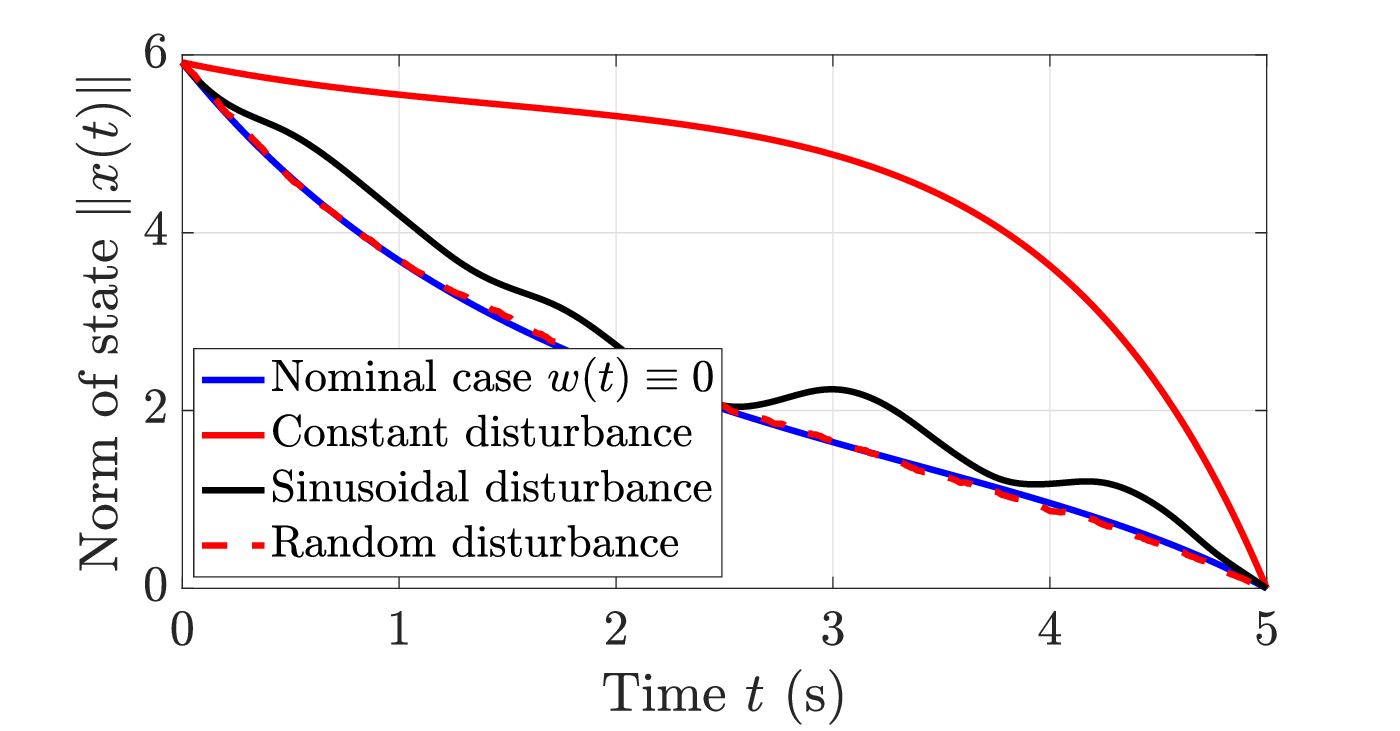}
	\caption{Norm of the state vector, $\|x(t)\|$ with time $t$ under the control laws \eqref{eq:Nominal_u} and \eqref{eq:Disturbed_u}. The states stabilize to the origin at $t_f = 5\mathrm{s}$.}
	\label{fig:states}
	\vspace{-0.3cm}
\end{figure}

\section{Simulation Example} \label{sec:Example}
In this section, we illustrate the use of our metrics on the ADMIRE fighter jet model \cite{SDRA}, a widely-used application for control frameworks \cite{Ola05, BO22b}. We consider only the subsystem associated with control actions, where the states are the roll, pitch and yaw rates, denoted $p$, $q$ and $r$ respectively, all in $\mathrm{rad/s}$. This subsystem has four inputs, corresponding to the deflections (in radians) of the canard wings, the left and right elevons and the rudder. A linearized model \eqref{eq:Nominal} for this subsystem was established in \cite{Ola05}, with:
$$
x = \begin{bmatrix} p \\ q \\ r \end{bmatrix}; \hspace{1em} A = \begin{bmatrix} -0.9967 & 0 & 0.6176 \\ 0 & -0.5057 & 0 \\ -0.0939 & 0 & -0.2127 \end{bmatrix};
$$
$$
B = \begin{bmatrix} 0 & -4.2423 & 4.2423 & 1.4871 \\ 1.6532 & -1.2735 & -1.2735 & 0.0024 \\ 0 & -0.2805 & 0.2805 & -0.8823 \end{bmatrix}.
$$
We note that the control inputs in the linearized model correspond to deviations from the equilibrium control inputs. In a practical setting, the actual control energies include the energy contributed by the equilibrium control inputs. In this example, we illustrate the use of our metrics by calculating control energies considering only the deviations.

\begin{figure}[!t]
	\centering
	\includegraphics[width = 0.4\textwidth]{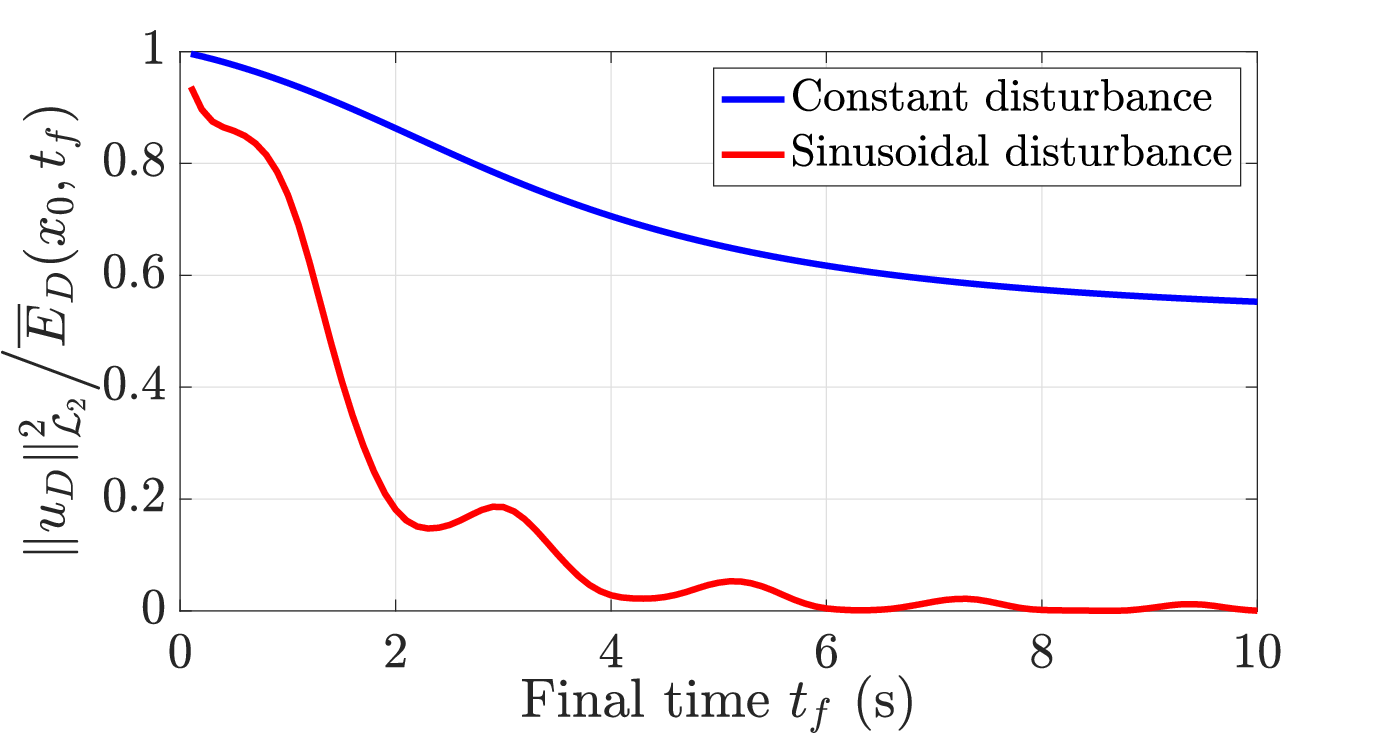}
	\caption{The ratio of $\|u_D\|_{\mathcal{L}_2}^{2}$ to $\overline{E}_{D}(x_0, t_f)$ in \eqref{eq:ED} as a function of the final time $t_f$, for two classes of disturbances.}
	\label{fig:uDvsED}
	\vspace{-0.5cm}
\end{figure}

We first illustrate the performance of the control laws \eqref{eq:Nominal_u} and \eqref{eq:Disturbed_u} in stabilizing this system. Note that these are defined only on the interval $[0, t_f]$. We use $x_0 = [5, -1, 3]^T$, $t_f = 5\mathrm{s}$ and $\overline{w} = 1$, and choose three classes of disturbances: (a) a constant, full-amplitude disturbance of $\pm\overline{w}$ in each input component, (b) different high-frequency sinusoids of amplitude $\overline{w}$ in each input component, and (c) a uniformly random disturbance $w(t) \in [-\overline{w}, \overline{w}]^3$. To illustrate the performance of the control law \eqref{eq:Disturbed_u}, we assume that each of these are known \emph{a priori}. Figure~\ref{fig:states} shows the norm of the state vector $\|x(t)\|$ as a function of time $t$ for different classes of disturbances. As expected, the states stabilize to the origin at the given final time. While the constant and sinusoidal disturbances significantly impact the trajectory of the state, we note that the uniformly random disturbance has a much smaller effect over the full time interval. 

Next, we demonstrate the accuracy of the bound \eqref{eq:ED} when the final time $t_f$ is varied. In particular, we consider the ratio of the actual energy $\|u_D\|_{\mathcal{L}_2}^{2}$ to the upper bound on the disturbed energy in \eqref{eq:ED}, denoted $\overline{E}_D(x_0, t_f)$. In Fig.~\ref{fig:uDvsED}, we plot this quantity as a function of the final time $t_f$. Note that $\|u_D\|_{\mathcal{L}_2}^{2} \leq \overline{E}_{D}(x_0, t_f)$ by definition, and this bound is more accurate when the ratio $\frac{\|u_D\|_{\mathcal{L}_2}^{2}}{\overline{E}_D(x_0, t_f)}$ is closer to $1$. We use the constant, full-amplitude disturbance and the high-frequency sinusoidal disturbance chosen in Fig.~\ref{fig:states}, with the same initial condition $x_0 = [5, -1, 3]^T$. It is evident that for small final times $t_f$, the bound is very accurate, and the accuracy is lowered as $t_f$ increases. Such small final times can correspond to quick maneuvering tasks performed by fighter jets. Further, the bound is more accurate for the constant disturbance compared to the sinusoidal disturbance, and this trend holds for other sets of initial conditions too.

\begin{figure}[!t]
	\centering
	\includegraphics[width = 0.4\textwidth]{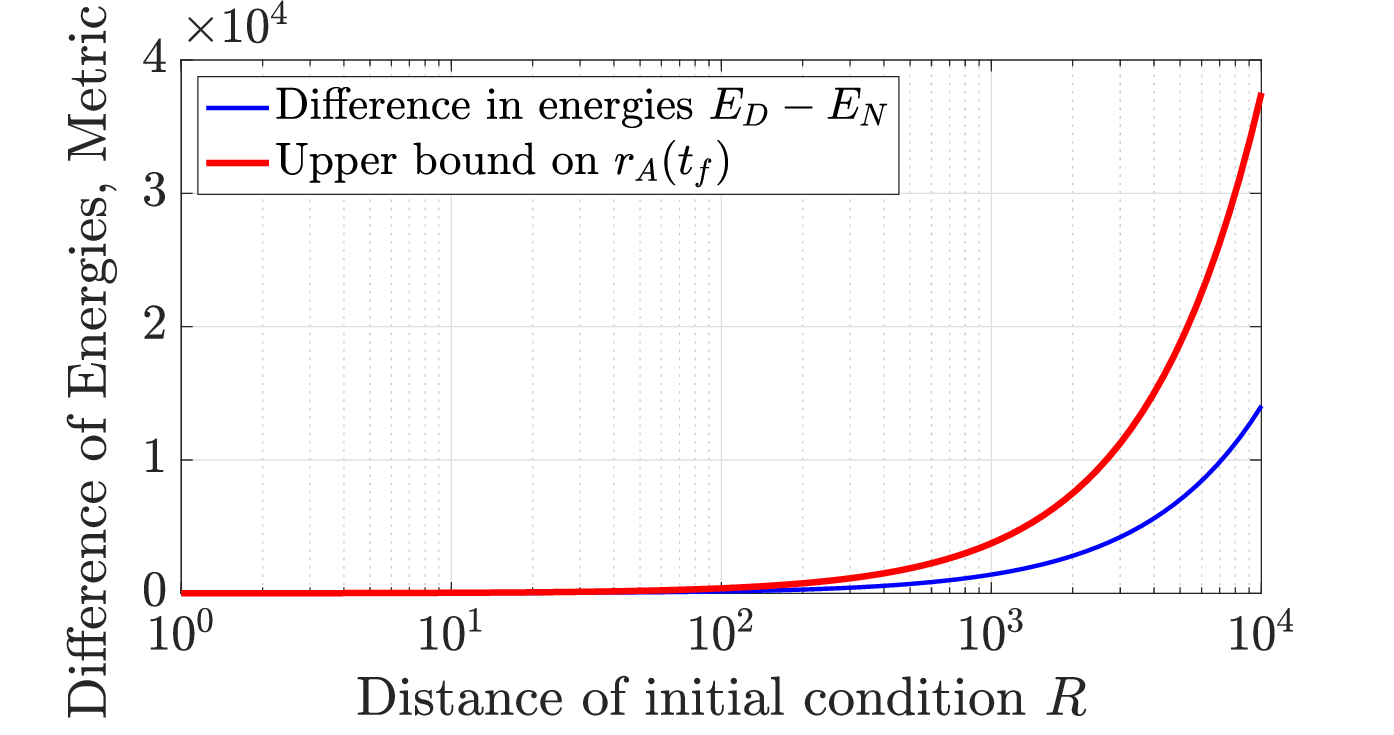}
	\caption{The bound on $r_A(t_f)$ from \eqref{eq:rA_bound}, and the difference in energies from \eqref{eq:EN} and \eqref{eq:ED}, as a function of the distance of the initial condition $R$.}
	\label{fig:rAvsR}
	\vspace{-0.3cm}
\end{figure}

\begin{figure}[!t]
	\centering
	\includegraphics[width = 0.4\textwidth]{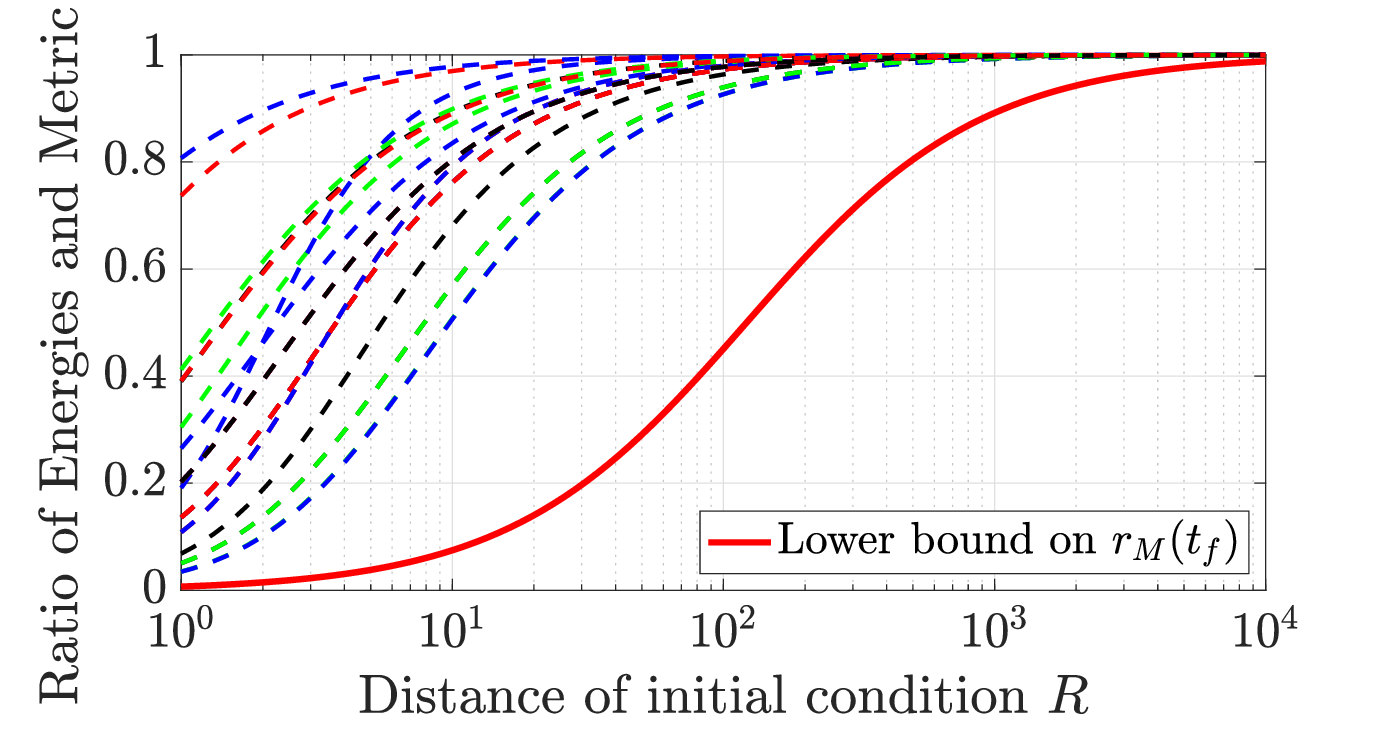}
	\caption{The bound on the $r_M(t_f)$ from \eqref{eq:rM_bound}, and the ratio of energies from \eqref{eq:EN} and \eqref{eq:uD2} for a large variety of disturbances, as a function of the distance of the initial condition $R$.}
	\label{fig:rMvsR}
	\vspace{-0.65cm}
\end{figure}

Figure \ref{fig:rAvsR} shows the difference between the nominal and disturbed energies from \eqref{eq:EN} and \eqref{eq:ED} as a function of the distance $R$ of the initial condition from the origin. Note that this quantity is independent of $w(t)$, and depends only on the bound $\overline{w}$. We fix $t_f = 0.5\mathrm{s}$ since Fig.~\ref{fig:uDvsED} shows that $\overline{E}_{D}(x_0, t_f)$ is a good approximation for $\|u_D\|_{\mathcal{L}_2}^{2}$ for such small $t_f$. For simplicity of notation, we denote $E_N \equiv E_{N}^{*}(x_0, t_f)$ and $E_D \equiv \overline{E}_D(x_0, t_f)$. We also plot the additive metric bound \eqref{eq:rA_bound} as a function of $R$, and note that $r_A(t_f)$ indeed bounds the difference in disturbed and nominal energies from above. Similarly, Fig.~\ref{fig:rMvsR} plots the ratio $\frac{E_{N}^{*}(x_0, t_f)}{\|u_D\|_{\mathcal{L}_2}^{2}}$ from \eqref{eq:EN} and \eqref{eq:uD2}, as a function of $R$ for a large variety of disturbances, fixing $t_f = 0.5\mathrm{s}$. These disturbances include different classes of high-frequency sinusoids and constant disturbances. We also plot the multiplicative metric bound \eqref{eq:rM_bound} as a function of $R$, and note that $r_M(t_f)$ bounds the ratio of nominal and disturbed energies from below. 

However, it can be seen that the bounds in Figures~\ref{fig:rAvsR} and~\ref{fig:rMvsR} are quite conservative. For instance, when $R = 10^2$ in Fig.~\ref{fig:rMvsR}, we have $r_M(t_f) = 0.45$ which indicates that \emph{at most} $1/0.45 \approx 2.22$ times the nominal energy is required to achieve finite-time stabilization, assuming optimal control signals are used. However, the actual ratio of energies is larger, ranging from $0.9$ to $1$, indicating that at most only $1/0.9 \approx 1.11$ times the nominal energy is required. A similar comment can be made from Fig.~\ref{fig:rAvsR} for the additive metric. We subsequently note that the additive metric is more representative when $R$ is small, and the multiplicative metric is more representative when $R$ is large.

\begin{figure}[!t]
	\centering
	\includegraphics[width = 0.4\textwidth]{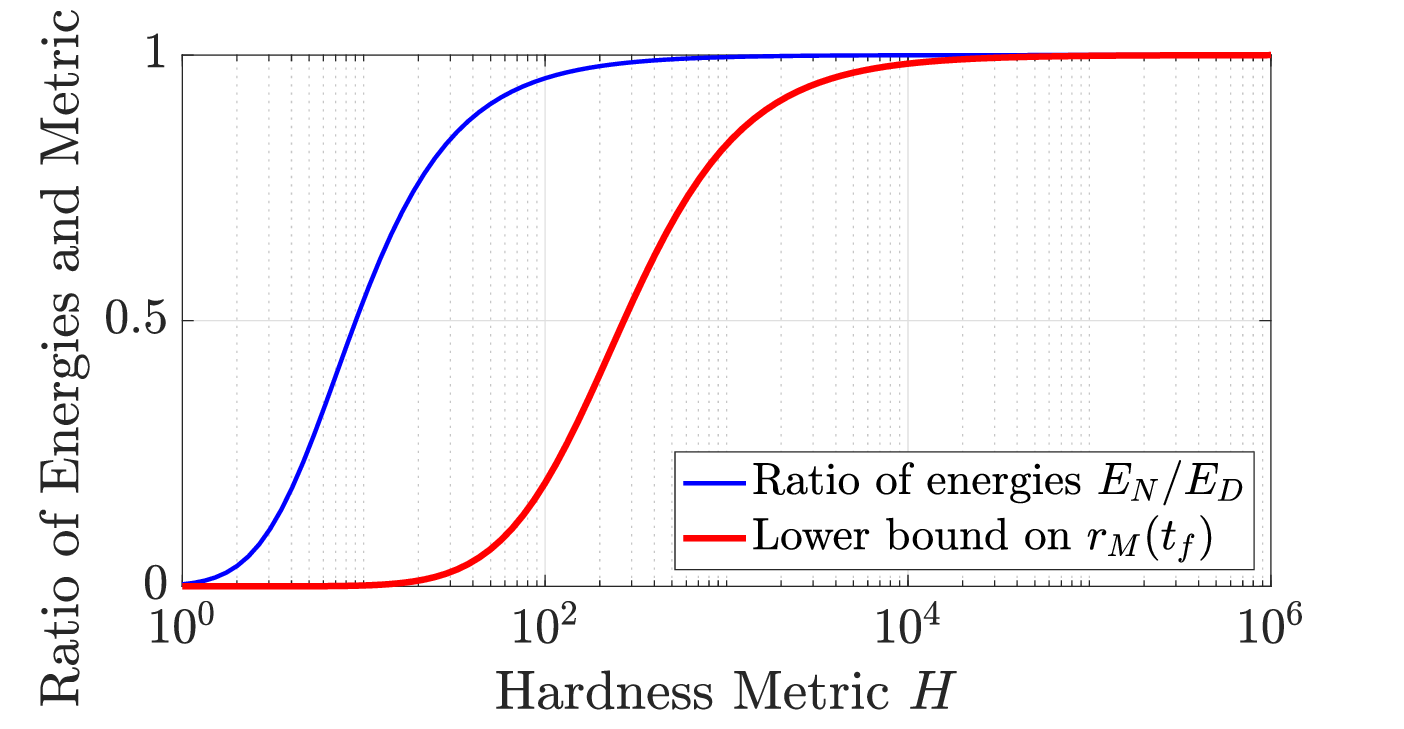}
	\caption{The multiplicative metric $r_M(t_f)$ from \eqref{eq:rM_bound} and the ratio of energies from \eqref{eq:EN} and \eqref{eq:ED} as a function of the hardness metric defined in \eqref{eq:Hardness}.}
	\label{fig:rMvsH}
	\vspace{-0.6cm}
\end{figure}

Finally, we define a measure of \emph{hardness} of a task:
\begin{equation} \label{eq:Hardness}
H = \frac{R}{t_f}.
\end{equation}
Larger values of $H$ intuitively imply that the task of finite-time stabilization is `harder', either due to a large distance of the initial condition $R$ or short final time $t_f$. Fig.~\ref{fig:rMvsH} illustrates the lower bound on $r_M(t_f)$ and the ratio $\frac{E_N}{E_D}$ as a function of this hardness metric. As $H$ increases, the bound on $r_M(t_f)$ is more accurate, indicating that the multiplicative metric is more representative for more difficult tasks which have either a large distance of the initial condition $R$ or short final time $t_f$. This behavior is closely related to the discussion of the multiplicative metric in Section \ref{sec:Metrics}.

\section{Concluding Remarks} \label{sec:Conclusion}
This paper introduced metrics to quantify the maximal additional energy required to achieve a given stabilization task in a finite time interval, in the presence of system disturbances. Motivated by problems of fuel capacity design in control systems, we first derive an upper bound on the worst-case energy over all disturbances. This quantity is compared to the nominal energy to achieve the same task in the absence of disturbances, using additive and multiplicative metrics. Simulation examples on a fighter jet model demonstrate that the metrics we define are practically useful particularly for tasks where the initial condition is far from the origin, or when the time interval is short. An important avenue for future work is to improve the tightness of the bounds we derive, in particular \eqref{eq:ED}, \eqref{eq:rA_bound} and \eqref{eq:rM_bound}. Characterizing the actual disturbance that maximizes the additional energy required is also an interesting future problem.

\balance

\section*{Acknowledgments}
The authors thank Elena Fern\'andez Bravo for useful comments and discussions.

\bibliographystyle{IEEEtran}
\bibliography{references}

\end{document}